\documentclass[aps,prl,twocolumn,superscriptaddress,floatfix,
nofootinbib,showpacs,longbibliography]{revtex4-1}

\usepackage[colorlinks=true, citecolor=blue, urlcolor=blue]{hyperref}
\usepackage[utf8]{inputenc}
\usepackage{graphicx}
\usepackage{subfigure}
\usepackage{setspace}
\usepackage{amsmath}
\usepackage{commath}
\usepackage{amsthm}
\usepackage{color}
\usepackage{braket}
\usepackage{color,soul}
\usepackage{appendix}
\usepackage{tgpagella}
\usepackage{amssymb}
\allowdisplaybreaks
\usepackage[dvipsnames]{xcolor}

\newtheorem{observation}{Observation}
\newtheorem{definition}{Definition}
\newtheorem{theorem}{Theorem}

\newtheorem{corollary}{Corollary}[theorem]
\newcommand{\oper}[1]{\ket{#1}\bra{#1}}

\begin{document}
\title{Locally unidentifiable subset of quantum states\\
and\\
its resourcefulness in secret password distribution}

\author{Pratik Ghosal}
\affiliation{Department of Physical Sciences, Bose Institute, EN 80, Sector V, Bidhan Nagar, Kolkata 700 091, India}

\author{Arkaprabha Ghosal}
\affiliation{Optics and Quantum Information Group, The Institute of Mathematical Sciences, CIT Campus, Taramani, Chennai 600113, India}

\author{Subhendu B. Ghosh}
\affiliation{Physics and Applied Mathematics Unit, Indian Statistical Institute, Kolkata, 203 B. T. Road, Kolkata 700108, India}

\author{Amit Mukherjee}
\affiliation{Indian Institute of Technology, Jodhpur-342030, India}

\begin{abstract}

    We introduce a hitherto unexplored form of quantum nonlocality, termed \textit{local subset unidentifiability}, that arises from the limitation of spatially separated parties to perfectly identify a \textit{subset} of mutually orthogonal multipartite quantum states, randomly chosen from a larger known set, using Local Operations and Classical Communication (LOCC). We show that this nonlocality is stronger than other existing forms of quantum nonlocality, such as local indistinguishability and local unmarkability. If more than one multipartite states from a locally indistinguishable set are distributed between spatially separated parties in a sequentially ordered fashion, then they may or may not mark which state is which using LOCC. However, we show that even when the parties cannot mark the states, they may still locally identify the particular states given to them, though not their order---\textit{i.e.}, they can identify the elements of the given subset of states. Then we prove the existence of such subsets that are not even locally identifiable, thereby manifesting a stronger nonlocality. We also present the genuine version of this nonlocality---\textit{genuine subset unidentifiability}---where the provided subset remains unidentifiable unless \textit{all} the parties come together in a common location and perform global measurements. We anticipate potential applications of this nonlocality for future quantum technologies. We discuss one such application in a certain \textit{secret password distribution protocol}, where this nonlocality outperforms its predecessors as a resource.
   

\end{abstract}

\maketitle
\textit{Introduction.}-  In any operational theory, if two or more
spatially separated communicating agents are able to
perform a distributed task then it is trivially executable
when they collaborate or come together in the same
lab. The converse statement is also true in classical
world. However, in the quantum regime, the scenario gets
more complex. There exist distributed quantum tasks
that can not be carried out locally, even though they
are perfectly accomplishable when the agents collaborate. This phenomenon exhibits a stark deviation of
quantum theory from the classical worldview.
Usually, this anomaly is termed as \textit{quantum nonlocality} \cite{peres,bennett1999quantum}. {This is different from the celebrated \textit{Bell-nonlocality} that deals with the incapability of explaining correlations by local-realism \cite{bell}.}
{Rather, this}
quantum nonlocality arises solely due to the intricate structure
and topology of the quantum state space. It has huge significance in both application \cite{dh1,dh2,dh3,dh4,ss} and theory. The indistinguishability of {certain} mutually orthogonal multipartite quantum states by local operations and classical communication (LOCC) is one of the basic quantum nonlocal phenomena. In the past decades, a plethora of intriguing results were developed along this line \cite{lid1,lid2,lid3,lid4,lid5,lid6,lid7,lid10,lid11,lid12,lid13,lid14,lid15,lid16,lid17,lid18,lid19,lid20,lid21,lid22,lid23,lid24,lid25,lid26,lid28,lid29,lid30,lid31,lid32,lid33,lid34,lid35,lid36,lid37,lid38,lid39,ghosh2022activating,gupta}. Very recently, some other distributed quantum tasks have also been developed that manifest \textit{stronger} forms of quantum nonlocality than the local indistinguishability \textit{viz.}, local state irreducibility \cite{halder}, local unmarkability \cite{sen} etc. As mentioned earlier, impossibility of locally realising some distributed tasks gives rise to these nonlocalities. These tasks can broadly be divided into two paradigms -- (i) \textit{local state discrimination} and, (ii) \textit{local state elimination}. Local indistinguishability and unmarkability are examples of the former, whereas local state irreducibility is for the later. In this work, we focus on the first paradigm and present an even stronger form of quantum nonlocality, exhibited by certain sets of mutually orthogonal entangled states. To demonstrate it we introduce a distributed task -- \textit{Local Subset Identification} (LSI). 
From a known set of mutually orthogonal states, {a subset of} \textit{more than one} state {is} chosen and shared among multiple spatially separated agents (see Fig. \ref{lsi}). The objective of the task is to perfectly identify the {subset, i.e., the identity of the} states {within it,} via LOCC. For example, take a set of bipartite orthogonal states. Consider that \textit{any two} states from the set are shared between spatially separated agents. Here, LSI asks {the agents} to perfectly locally  identify which \textit{two} states are given to them using LOCC. {We show that the inability to perfectly accomplish this task for a set demonstrates a stronger form of quantum nonlocality than those discussed in the existing literature, termed \textit{Local Subset Unidentifiability}.}  
At the core of this proposed nonlocality lies the inability to perfectly distinguish certain sets of subspaces of rank more than one using LOCC. Since the task of addressing (in)distinguishability of sets of subspaces is naturally more complex than that of sets of vectors, our proposed nonlocality stands out compared to the previous nonlocalities. Moreover, we also deliver an information processing application of the proposed nonlocality.

Dispersal of information to maintain its secrecy is crucial in information processing. Our proposed nonlocality also provides an interesting application in this line. Local subset unidentifiability furnishes heightened security in some secret password sharing scheme. We demonstrate that other nonlocal resources \textit{viz.} local unmarkability if used in this task keeps the possiblity of significant information leakage, whereas local unidentifiability {gives significant advantage over local unmarkability}. This makes local unidentifiability potentially useful in distributed protocols to support the emerging quantum internet technology \cite{inta,intb}.   

Furthermore, any nonlocal feature of quantum systems gets more intricate when multipartite scenario comes into the picture. In case of LSI, we also explore scenarios involving more than two spatially separated agents. Interestingly, we come up with a further stronger version of the nonlocality we are introducing here. In particular, we present sets of multipartite states that shows local subset unidentifiability (or \textit{locally unidentifiable}) when all the agents are spatially separated. We show that these sets retain local unidentifiablity in all possible bi-partitions. Therefore, to perfectly accomplish the LSI task, \textit{all} the agents need to come together or must resort to additional quantum resources. We term it as \textit{Genuine Unidentifiablity}. We also illustrate that any set that is genuinely locally unidentifible must also show \textit{genuine}  unmarkability which is hitherto a uncharted notion.

\textit{Local Subset Identification.}- 
We are now in a position to articulate the formal definition of our task (see Fig. \ref{lsi}). 
\begin{definition}[$(n,\abs{\mathcal{S}'})$-Local Subset Identification]
Consider a set $\mathcal{S}=\lbrace \ket{\psi_1},\ket{\psi_2},\ldots,\ket{\psi_n} \rbrace \subset \bigotimes_{k=1}^N \mathbb{C}^{d_k}$ of $n$ $N$-party orthogonal quantum states. A subset $\mathcal{S}'\subset \mathcal{S}$ containing  $1\le \abs{\mathcal{S}'}<n$ quantum states is randomly chosen and distributed among spatially separated {classically communicating} agents keeping its identity hidden. The task of $(n,\abs{\mathcal{S}'})$-local identification is to perfectly determine the elements of the set $\mathcal{S}'$.
\end{definition}

We call a set $\mathcal{S}$ to be {\textit{locally} $(n,\abs{\mathcal{S}'})$\textit{-identifiable} or simply $(n,\abs{\mathcal{S}'})$\textit{-identifiable}} if \textit{all possible} subsets of $\mathcal{S}$ with the same cardinality as $\abs{\mathcal{S}'}$ are locally identifiable, otherwise we call it {\textit{locally} $(n,\abs{\mathcal{S}'})$\textit{-unidentifiable} or simply $(n,\abs{\mathcal{S}'})$\textit{-unidentifiable}}. It is evident that $\abs{\mathcal{S}'}=n$ is a trivial question in the LSI framework, {as the elements of the entire set $\mathcal{S}$ is already known to the agents}. On the other extreme, when $\abs{\mathcal{S}'}=1$, this boils down to a well known task:
\begin{figure}[t]
	\centering
	\includegraphics[width=3.5in]{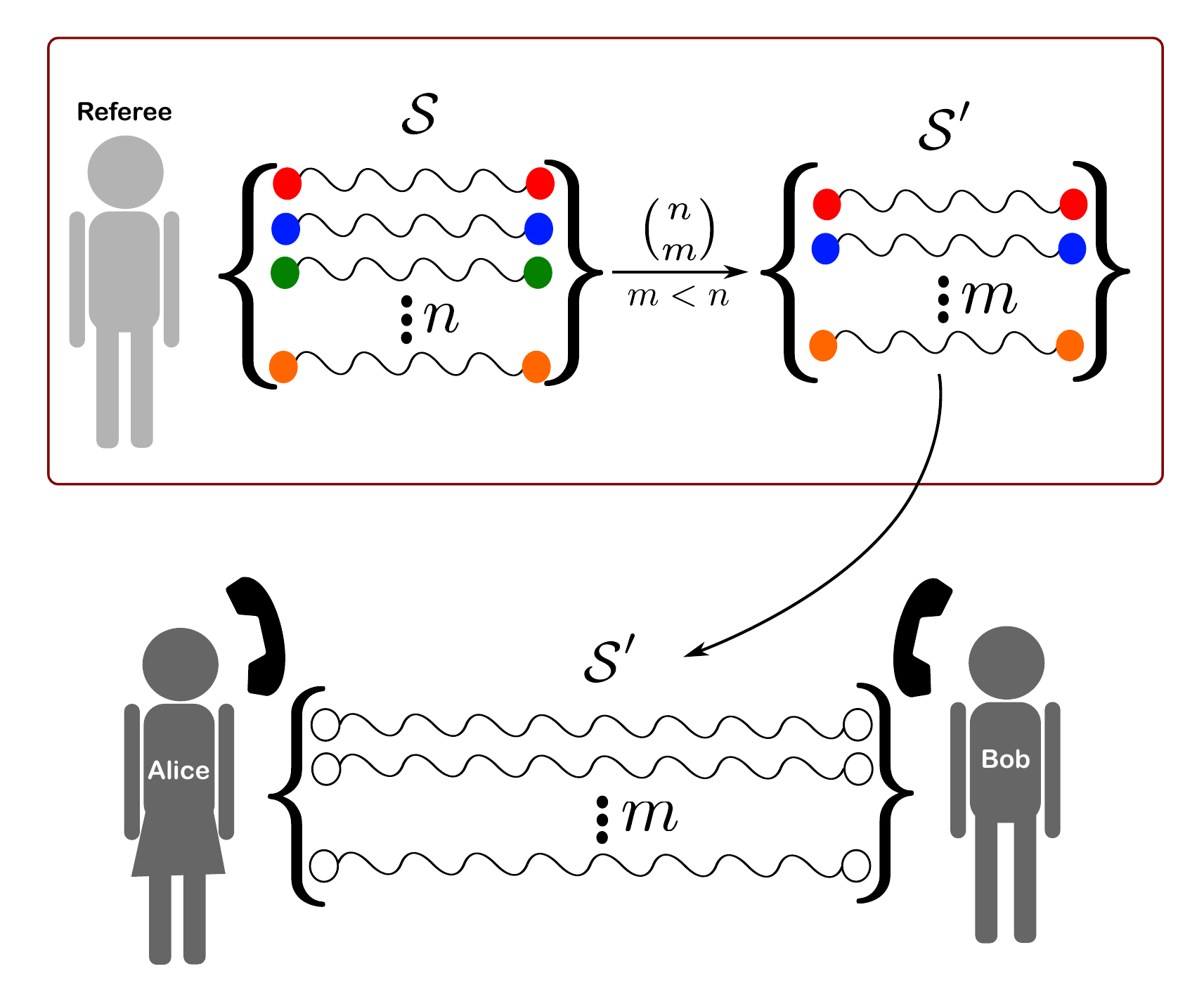}
	\caption{Local Subset Identification. Referee chooses $m$ states from a set of $n$ multipartite states, denoted by $\mathcal{S}$. The newly chosen set is denoted as $\mathcal{S}'$. The set $\mathcal{S}'$ is then shared between the spatially separated agents, say, Alice and Bob who do not know the constituent elements of $\mathcal{S}'$. The main task of the agents is to perfectly identify the elements of $\mathcal{S}'$ with the help of LOCC only. However, the order of the elements is not required.}
	\label{lsi}
\end{figure}
{
 \begin{observation}
{For a given set $\mathcal{S}$ the task of $(n,1)$-LSI corresponds to the well-known task of Local State Discrimination (LSD) \cite{ld}.}
\end{observation}}

Another related distributed task is {Local State Marking} (LSM) \cite{sen}. For the sake of completeness and comparison, we briefly describe it here. {Consider the set $\mathcal{S}=\{\ket{\psi_i}_{\mathtt{AB}}\}_{i=1}^n$ of $n$ bipartite orthogonal states from which a subset $\mathcal{S}'$ is shared between two spatially separated agents. Whereas the task of LSI is to locally identify \textit{just} the elements of the subset, the task of $(n,\abs{\mathcal{S}'})$-LSM is to perfectly determine the \textit{order as well} in which the elements are given to the agents. For example, consider $\mathcal{S}'=\{\ket{\psi_j},\ket{\psi_k}\}$ for any $j<k\in\{1,\cdots,n\}$ \cite{setorder}. These two states can be shared between the agents in two distinct orders: $\ket{\psi_j}_{\mathtt{A_1B_1}}\otimes\ket{\psi_k}_{\mathtt{A_2B_2}}$ and $\ket{\psi_k}_{\mathtt{A_1B_1}}\otimes\ket{\psi_j}_{\mathtt{A_2B_2}}$, both representing the same subset $\mathcal{S}'$. LSI asks the agents just the identitiy of the two states $\ket{\psi_j}$ and $\ket{\psi_k}$. LSM additionally asks whether they share $\ket{\psi_j}_{\mathtt{A_1B_1}}\otimes\ket{\psi_k}_{\mathtt{A_2B_2}}$ or $\ket{\psi_k}_{\mathtt{A_1B_1}}\otimes\ket{\psi_j}_{\mathtt{A_2B_2}}$. Thus, the LSM task can be formulated as a higher-dimensional LSD problem of the set $\{\ket{\psi_i}_{\mathtt{A_1B_1}}\otimes\ket{\psi_{i'}}_{\mathtt{A_2B_2}}\}_{i\neq i'=1}^n$ containing $^nP_2$ states. This formalism can be extended to accommodate any number of parties and subsets with larger cardinalities.}
If for a specific set of states this distributed task can not be accomplished perfectly locally, the set is said to {exhibit} a stronger quantum nonlocality than the local indistinguishability known as local unmarkability. If the $(n,\abs{\mathcal{S}'})$-LSM task can be perfectly accomplished, we call the set $\mathcal{S}$ as $(n,\abs{\mathcal{S}'})$-markable, otherwise it is $(n,\abs{\mathcal{S}'})$-unmarkable. 
Clearly, LSI is an easier task than LSM. {Hence}, local unidentifiability is a stronger nonlocal phenomenon than local unmarkability as well as local indistinguishability \cite{compare}. 
  
{\begin{observation}
    If a set $\mathcal{S}$ is $(n,m)$-markable, then it readily follows that it is also $(n,m)$-identifiable.
\end{observation}}
{If the $N$ spatially separated agents can locally determine the order of the states of the given subset $\mathcal{S'}$, then the identity of the states { becomes} trivially known to them.} However, the converse of this statement is not necessarily true. 
\begin{theorem}\label{32}
If a set $\mathcal{S}$ is $(n,m)$-identifiable, then it is not necessarily $(n,m)$-markable.
\end{theorem}
\begin{proof}
{The proof is constructive. Consider a set of three Bell states: $\mathcal{S}=\lbrace \ket{B_i}_{\mathtt{AB}}\rbrace_{i=1}^3$, where $\ket{B_1}_{\mathtt{AB}}=\frac{1}{\sqrt{2}}(\ket{00}_{\mathtt{AB}}+\ket{11}_{\mathtt{AB}}), \ket{B_2}_{\mathtt{AB}}=\frac{1}{\sqrt{2}}(\ket{00}_{\mathtt{AB}}-\ket{11}_{\mathtt{AB}}), \ket{B_3}_{\mathtt{AB}}=\frac{1}{\sqrt{2}}(\ket{01}_{\mathtt{AB}}+\ket{10}_{\mathtt{AB}})$. This set is $(3,2)-$identifiable, but $(3,2)-$unmarkable.} 

{\textit{$(3,2)-$identifiability:} A subset containing two distinct Bell states from the set $\mathcal{S}$ can be chosen in ${3\choose2}=3$ ways, namely, $\mathcal{S}'_1 = \{\ket{B_1},\ket{B_2}\}, \mathcal{S}'_2= \{\ket{B_2},\ket{B_3}\}, \mathcal{S}'_3=\{\ket{B_3},\ket{B_1}\}$. One such subset is randomly selected and the states within it are distributed among two spatially separated agents, Alice and Bob (see Fig. \ref{lsi}). They can identify the elements of the given subset by performing the following measurement $\mathcal{M}_\alpha:=\{P^\alpha_i|\sum_{i=1}^4 P^\alpha_i=\mathbb{I}_4\}$, where $P^\alpha_i:=\ket{B_i}_\alpha\bra{B_i}$, $\alpha\in\{\mathtt{A}_1\mathtt{A}_2,\mathtt{B}_1\mathtt{B}_2\}$ and $\ket{B_4}=\frac{1}{\sqrt{2}}(\ket{01}-\ket{10})$, on their respective share of two qubits, and comparing their measurement outcomes via classical communication. See Appendix A for details. Similar result holds for any set of three of Bell states.}

{\textit{$(3,2)-$unmarkability:} However, Alice and Bob can never locally mark (i.e., distinguish between the orders of) two states from any set of three Bell states. This is because the task involves locally distinguishing a state from a set of $^3P_2=6$ maximally entangled states (MESs) $\in \mathbb{C}^4\otimes\mathbb{C}^4$, a scenario deemed impossible by Ref. \cite{Hayashi}. The reference asserts that perfect LSD of MESs becomes unattainable when the number of states exceeds the local dimension of the shared state.}
\end{proof}

{In Appendix A, we present additional examples of locally identifiable yet unmarkable sets. It's worth noting that if a set $\mathcal{S}$ is $(n,m)$-\textit{unidentifiable}, it is inherently $(n,m)$-\textit{unmarkable}, highlighting local subset unidentifiability as a stronger form of nonlocality. Here are some examples of locally unidentifiable sets.}

\begin{theorem}\label{4uid}
The set of four two-qubit Bell states is (4,2)-unidentifiable.
\end{theorem}
\begin{proof} If a subset $\mathcal{S}'_i$ containing two states from a set $\mathcal{S}$ of four two-qubit Bell states is shared among two spatially separated agents, Alice and Bob, they will {never} be able to locally distinguish between the $\mathcal{S}'_i$s, {\textit{i.e.}, know the identity of the two given Bell states, let alone their ordering}. Following is the detailed proof. 

Consider that $\mathcal{S}:=\{\ket{B_1}, \ket{B_2}, \ket{B_3}, \ket{B_4}$\}. A subset containing any two Bell states can be chosen in {${4\choose2}=6$} ways, namely, $\mathcal{S}'_1:=\{\ket{B_1},\ket{B_2}\}$, $\mathcal{S}'_2:=\{\ket{B_1},\ket{B_3}\}$, $\mathcal{S}'_3:=\{\ket{B_1},\ket{B_4}\}$, $\mathcal{S}'_4:=\{\ket{B_2},\ket{B_3}\}$, $\mathcal{S}'_5:=\{\ket{B_2},\ket{B_4}\}$, $\mathcal{S}'_6:=\{\ket{B_3},\ket{B_4}\}$. 
Now, if $\mathcal{S}'_1$ is distributed among the agents, {the shared states can exist in two distinct orders, either $\ket{B_1}_{\mathtt{A_1B_1}}\otimes\ket{B_2}_{\mathtt{A_2B_2}}$ or $\ket{B_2}_{\mathtt{A_1B_1}}\otimes\ket{B_1}_{\mathtt{A_2B_2}}$, both with equal probability, representing the same subset $\mathcal{S}'_1$. Consequently, the description of the subset according to the agents corresponds to a mixed state: }
{\small
\begin{equation*}
    \begin{split}
    \rho_1:=\tfrac{1}{2}\oper{B_1}_{\mathtt{A_1B_1}}&\otimes\oper{B_2}_{\mathtt{A_2B_2}} \\& + \tfrac{1}{2}\oper{B_2}_{\mathtt{A_1B_1}}\otimes\oper{B_1}_{\mathtt{A_2B_2}}.
    \end{split}
\end{equation*}}  
{More generally, Alice and Bob's description of  subset $\mathcal{S}'_i=\{\ket{B_{\alpha}},\ket{B_{\beta}}\}$ is given by}
{\small
\begin{equation*}
    \begin{split}
    \rho_i:=\tfrac{1}{2}\oper{B_{\alpha}}_{\mathtt{A_1B_1}}&\otimes\oper{B_{\beta}}_{\mathtt{A_2B_2}} \\& + \tfrac{1}{2}\oper{B_{\beta}}_{\mathtt{A_1B_1}}\otimes\oper{B_{\alpha}}_{\mathtt{A_2B_2}},
    \end{split}
\end{equation*}}
for $i\in\{1,\cdots,6\}$.

Clearly, the task of locally distinguishing the subsets $\{\mathcal{S}'_i\}_{i=1}^6$ is equivalent to distinguishing between the mixed states $\{\rho_i\}_{i=1}^6$ by LOCC. We now prove our claim by \textit{reductio ad absurdum}. 

We assume  the states $\{\rho_i\}_{i=1}^6$ are locally distinguishable. Consequently, their supports are also distinguishable via LOCC \cite{sb}. Hence, any set of pure states {$\{\ket{\xi_i}\left\vert\right.\ket{\xi_i}\in \text{Support}(\rho_i)\}_{i=1}^6$} must also be locally distinguishable.  However,  the set $\{\ket{B_j}_{\mathtt{A_1B_1}}\otimes\ket{B_k}_{\mathtt{A_2B_2}}\}_{j< k=1}^4$, each state of which $\in \text{Support}(\rho_i)$, is locally indistinguishable as the number of MESs exceeds their local dimension \cite{Hayashi}.  This readily contradicts our assumption that $\{\rho_i\}_{i=1}^6$ are locally distinguishable,  thus concluding our proof. 
\end{proof}

{We have also identified several sets containing orthogonal general MESs that fail to achieve the LSI task, thus demonstrating local subset unidentifiability. 
}

\begin{theorem}\label{guid}
{Consider a set $\mathcal{S}$ of $2<D\leq d^2$ MESs in $\mathbb{C}^d \otimes \mathbb{C}^d$. There are $D \choose k$ possible subsets, each containing $1<k<D$ distinct states from $\mathcal{S}$. The set $\mathcal{S}$ is ($D,k$)-unidentifiable, if ${D \choose k}>d^k$.}
\end{theorem}
The proof of this is provided in Appendix B. 

{Now, as a corollary of the above theorem, we present the following result, which is in spirit of the theorem stating that any complete basis set of MESs is locally indistinguishable \cite{Hayashi}.}
\begin{corollary}
{Any two distinct states selected from a complete basis set of MESs of arbitrary dimensions cannot be locally identified.}
\end{corollary}
\begin{proof}
{Consider a complete basis set of $d^2$ mutually orthogonal MESs in $\mathbb{C}^{d}\otimes\mathbb{C}^{d}$. From this set, any two states can be chosen in ${d^2 \choose 2}=\frac{d^2(d^2-1)}{2}$ ways. This quantity is always greater than $d^2$ for any $d \in \mathbb{Z}^+\backslash\{1\}$. Therefore, the subset of two states is always locally unidentifiable.}
\end{proof}

{
Now, we extend our analysis to the multipartite framework. Similar extensions exist for other known forms of nonlocality, such as local indistinguishability and local irreducibility \cite{halder,rout2019genuinely,rout2021multiparty}.}

\textit{Genuine Unindetifiability.}- 
{We begin by presenting an example of set of tripartite orthogonal states demonstrating local unidentifiability in every bipartition. We show that even when \textit{all but one} agent collaborates in a single laboratory, with the freedom to perform the strongest possible joint operations on their composite subsystems, they still fail to achieve perfect LSI. We term this multipartite notion of unidentifiablity as \textit{genuine}. Furthermore, our examples also represent genuine multipartite instances of the existing concept of local unmarkability \cite{sen}.}
{Consider a set of eight $3$-qubit GHZ states, $\mathcal{S}:=\{\frac{1}{\sqrt{2}}(\ket{000}\pm\ket{111}),\frac{1}{\sqrt{2}}(\ket{001}\pm\ket{110}),\frac{1}{\sqrt{2}}(\ket{010}\pm\ket{101}),\frac{1}{\sqrt{2}}(\ket{100}\pm\ket{011})\}$.} 

\begin{theorem}\label{gen}
The set $\mathcal{S}$ is locally $(D,k)$-unidentifiable if ${D\choose k}>2^{2k}${, where
$2<D\leq 8$ and $1<k<D$,} 
even when two out of three parties { collaborate} in a same lab.  
\end{theorem}

{Proof of this is given in Appendix C.} 

The scenario gets further involved with increasing number of parties. {For  instance,  when we add another spatially separated agent to the tripartite scenario, we encounter two distinct types of bi-partitions:} 1 party vs. 3 parties ($1:3$), and 2 parties vs. 2 parties ($2:2$). {Within each type, there exist bi-partitions} that are related by party permutations. Sets of orthogonal four-party states that exhibit local unidentifiability in {\it all} such possible bi-partitions {within both types} bring out the true genuineness of the nonlocality discussed here.

Consider the set $\mathcal{S}:=\{\ket{\Omega_{\alpha}}\}_{\alpha=1}^{16}$, where
{\small
\begin{align}
    \ket{\Omega_{1,2}} &= \tfrac{1}{2} (\ket{0000}\pm\ket{0111}+\ket{1010}\pm\ket{1101}),\nonumber\\
    \ket{\Omega_{3,4}} &= \tfrac{1}{2} (\ket{0000}\pm\ket{0111}-\ket{1010}\mp\ket{1101}),\nonumber\\
    \ket{\Omega_{5,6}} &= \tfrac{1}{2} (\ket{0001}\pm\ket{0110}+\ket{1011}\pm\ket{1100}),\nonumber\\
    \ket{\Omega_{7,8}} &= \tfrac{1}{2} (\ket{0001}\pm\ket{0110}-\ket{1011}\mp\ket{1100}),\nonumber\\
    \ket{\Omega_{9,10}} &= \tfrac{1}{2} (\ket{0010}\pm\ket{0101}+\ket{1000}\pm\ket{1111}),\nonumber\\
    \ket{\Omega_{11,12}} &= \tfrac{1}{2} (\ket{0010}\pm\ket{0101}-\ket{1000}\mp\ket{1111}),\nonumber\\
    \ket{\Omega_{13,14}} &= \tfrac{1}{2} (\ket{0011}\pm\ket{0100}+\ket{1001}\pm\ket{1110}),\nonumber\\
    \ket{\Omega_{15,16}} &= \tfrac{1}{2} (\ket{0011}\pm\ket{0100}-\ket{1001}\mp\ket{1110}).\nonumber
\end{align}}

\begin{theorem}\label{gen4}
The set $\mathcal{S}$ is locally $(D,k)$-unidentifiable if ${D\choose k}>2^{3k}${, where $2<D\leq16$ and $1<k<D$,} in all $1:3$ as well $2:2$ bi-partitions.
\end{theorem}

The proof of this theorem is provided in Appendix D.
Additionally, Theorem \ref{gen} and \ref{gen4} also present examples of \textit{genuine unmarkability}. The concerned sets are locally unmarkable in all possible bipartitions. 

\textit{Application in secret password distribution.}- 
{Our notion of stronger nonlocality, Local Subset Unidentifiability, has potential significance in various information processing technologies. In this context, we focus on a specific application (see Fig. \ref{apl}).
\begin{figure}[t]
	\centering
	\includegraphics[width=3.5in]{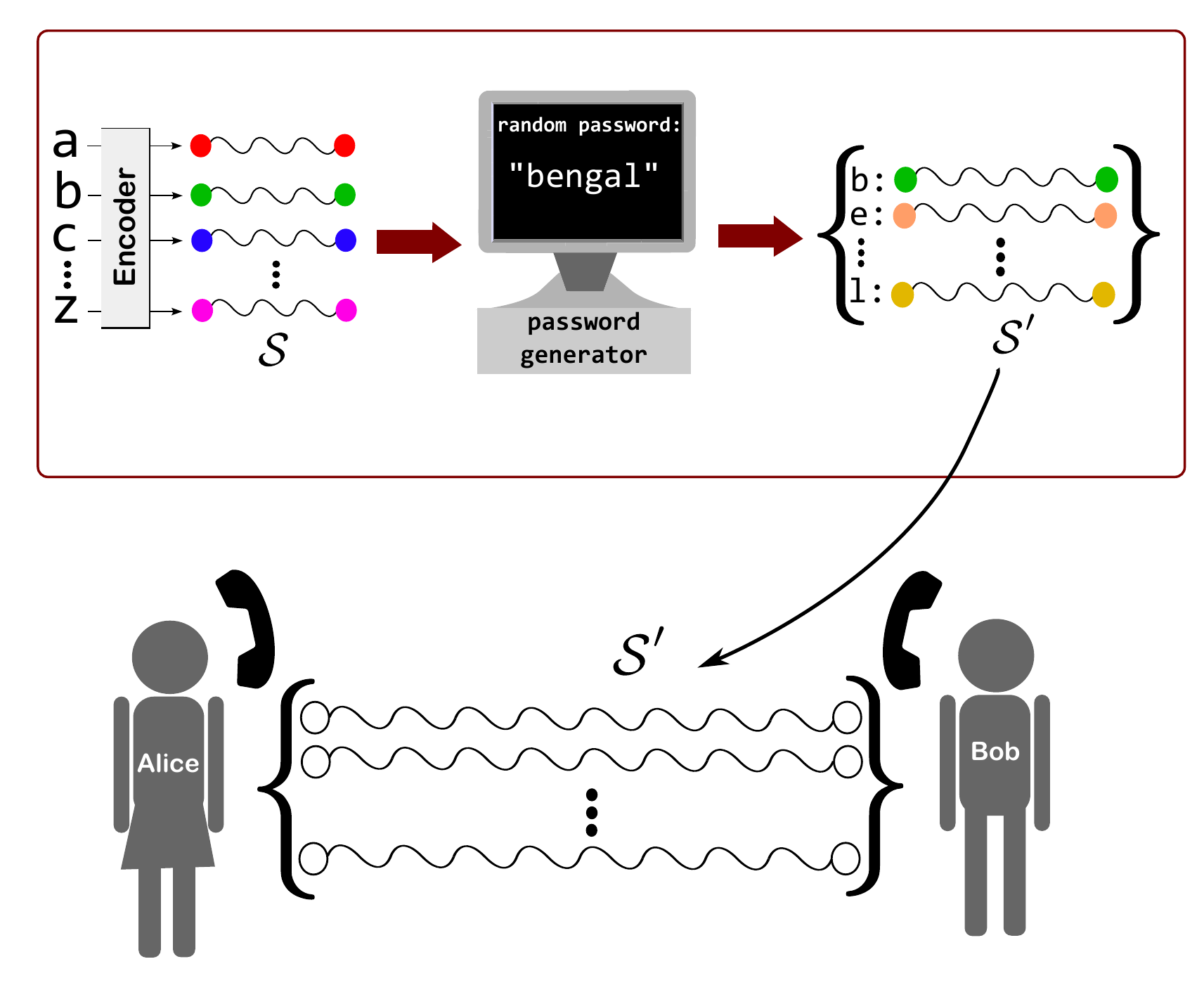}
    \caption{{
    Improving the information leakage in secret password distribution scheme. A random password-string of $m$ letters is generated from a known alphabet of $n>m$ letters, to be distributed among spatially separated parties. The sender wishes to keep the identity of the password hidden unless all the parties physically come together in a common location. If the letters of the alphabet are encoded into a set of multipartite orthogonal states such that the set is locally $(n,m)$-unidentifiable, then the spatially separated parties cannot locally reveal the hidden password (not even the letters) with certainty. This provides a tighter security condition than encoding the letters into a locally unmarkable set, where the receivers may locally predict the identity of each letter in the password with certainty.
    }   }
	\label{apl}
\end{figure}
In this scenario, a sender aims to distribute hidden information, such as a \textit{locked password}, among several spatially separated receivers who can communicate classically among themselves. The sender's objective is to ensure that the receivers remain unaware of the password's identity as long as they are separated, even if they have unrestricted classical communication. The password can only be unlocked when all the receivers physically gather in a common location.}

Precisely, consider that the sender wish to share a password -- a string $\mathcal{X}:=x_1x_2\cdots x_m$ of $m$ letters, $x_i$ being the $i$th letter in the string -- among the receivers. Each letter in the string is to be chosen without repetition from an alphabet $\mathcal{A}=\{a_k\}_{k=1}^n$ of $n>m$ letters which is known to the receivers as well.
Now, the sender and the receivers agree upon an encoding scheme: the letters of the alphabet $\mathcal{A}$ are encoded in a set $\mathcal{S}:=\{\ket{\psi_k}\}_{k=1}^n$ of pairwise orthogonal pure multipartite quantum states. Accordingly, the sender encodes their password $\mathcal{X}$ into a string of quantum states:
$\mathcal{X}\mapsto \ket{\xi_1} \otimes \ket{\xi_2} \otimes \cdots \otimes \ket{\xi_m}$, where the state $\ket{\xi_i}$ can be any state from $\mathcal{S}$ with the only restriction that $\ket{\xi_i}\neq \ket{\xi_j}, \forall i,j$. 
Subsequently, the sender shares this composite state among the receivers (see Fig. \ref{apl}). If $\mathcal{S}$ is locally ($n,m$)-unmarkable, then the spatially separated parties will not be able to perfectly discriminate the received string from the $^nP_m=\frac{n!}{(n-m)!}$ possible strings of quantum states by LOCC. However, it may so happen that the receivers can perfectly predict the identity of the $m$ individual states $\lbrace \ket{\xi_i}\rbrace$ if the encoding is done in a locally $(n,m)-$unmarkable, but locally $(n,m)-$identifiable set of quantum states. Then, they will be able to guess the correct permutation of the letters with success probability $\frac{1}{m!}$. Another alternative which they may opt for is to imperfectly discriminate (\textit{i.e.}, with a nonzero probability $P_{imp}<1$, bounded by an upper limit discussed in \cite{imperfect1}) the received string. The encoding set $\mathcal{S}$ determines which of the success probabilities is higher.

{However,} if the sender encodes the password in a locally ($n,m$)-unidentifiable set of states, then the receivers will not be able to even identify the individual letters perfectly by LOCC, and hence they will not follow the former strategy. Furthermore, if, for such sets, $P_{imp}<\frac{1}{m!}$, then the security of the password is enhanced significantly. 
We have found examples of $(n,m)-$unidentifiable sets, $\mathcal{S}_d\subset\mathbb{C}^d\otimes\mathbb{C}^d$ containing {MESs} for which $P_{imp}<\frac{1}{m!}$ (follows from Theorem \ref{guid}) as long as $|\mathcal{S}_d|\geq d+1$ \cite{imperfect1,imperfect2}.

\textit{Discussion.}- 
In summary, we come up with a new distributed task -- LSI. We show that impossibility of accomplishing this task gives birth to a unique version of quantum nonlocality -- local subset unidentifiability. We show that this is by-far the strongest quantum nonlocality in the state discrimination paradigm, that arises from the impossibility of discriminating certain mutually orthogonal subspaces of rank more than one. In the multipartite framework, we introduce the notion of genuine unidentifiability which says that a set of quantum states may remain locally unidentifiable even in all possible bi-partitions. Along this line, we also introduce the notion of genuine unmarkability in multipartite scenarios. Interestingly, we also propose a cryptographic application of this proposed nonlocality. In secret password distribution scheme, we demonstrate that local unidentifiability provides a strictly better encoding for protecting password secrecy than its predecessors. While we explore local unidentifiability only in entangled states, we believe it is not necessarily an exclusive characteristic of entanglement. Exploring the same feature in orthogonal product states would be quite intriguing.

\begin{acknowledgements}
We thank Guruprasad Kar, Somshubhro Bandyopadhyay, Sibasish Ghosh, Manik Banik, and Arup Roy for discussions. 
\end{acknowledgements}


\begin{thebibliography}{60}
\bibitem [{ld()}]{ld}%
  \BibitemOpen
  \href@noop {} {}\bibinfo {note} {Local state discrimination (LSD) is a {well
  known distributed task} where one state randomly chosen from a known set
  $\mathcal{S}$ of orthogonal multipartite states is shared among spatially
  separated agents and their objective is to determine the identity of the
  given state by LOCC.}\BibitemShut {Stop}%
\bibitem [{set()}]{setorder}%
  \BibitemOpen
  \href@noop {} {}\bibinfo {note} {Since both $\{\ket{\psi_j},\ket{\psi_k}\}$
  and $\{\ket{\psi_k},\ket{\psi_j}\}$ represent the same set, we take $j<k$ to
  avoid double counting.}\BibitemShut {Stop}%
\bibitem [{com()}]{compare}%
  \BibitemOpen
  \href@noop {} {}\bibinfo {note} {Note that if a set $\mathcal{S}$ is locally
  distinguishable then obviously the set is locally markable as well as locally
  {identifiable.} Also, if the set $\mathcal{S}$ is $(n,m)$-unidentifiable, for
  $n>m>1$, it must be $(n,m)$-unmarkable. However, $(n,m)$-LSI or $(n,m)$-LSM
  task is not comparable in straight-forward sense to the LSD task for $m>1$ as
  the {last one} is defined only for $m=1$. Nevertheless, it is obvious that if
  a set is $(n,m)$-unidentifiable or $(n,m)$-unmarkable it must be locally
  indistinguishable.}\BibitemShut {Stop}%
\end{thebibliography}


%

\begin{appendix}
\onecolumngrid
    \section{Appendix A: Locally identifiable subsets of Two-qubit Bell states}
{
\textbf{$(3,2)-$ Bell identifiability: }
Alice and Bob can successfully identify any subset of two Bell states $\{\ket{B_i}, \ket{B_j}\}_{i\neq j }$ given from a known set of three Bell states, $\{\ket{B_i}, \ket{B_j}, \ket{B_k}\}_{i\neq j\neq k}$ by LOCC only, where
\begin{align}
    \ket{B_{1,2}}=\frac{1}{\sqrt{2}}(\ket{00}\pm \ket{11}), \quad \quad \ket{B_{3,4}}=\frac{1}{\sqrt{2}}(\ket{01}\pm \ket{10}) \nonumber
\end{align}

The protocol is the following. Each agent (Alice or Bob) performs a Bell basis measurement $\mathcal{M}_\alpha:=\{P^\alpha_i|\sum_{i=1}^4 P^\alpha_i=\mathbb{I}_4\}$, where $P^\alpha_i:=\ket{B_i}_\alpha\bra{B_i}$, $\alpha\in\{\mathtt{A}_1\mathtt{A}_2,\mathtt{B}_1\mathtt{B}_2\}$, on the two qubits in their respective labs. Afterwards, the agents can know the identity of the shared subset by tallying their measurement outcomes via classical communication as the following table represents,
\begin{center}
\begin{table}[h]
\begin{tabular}{|cl|cl|cl|}
\hline
\multicolumn{2}{|c|}{$\mathcal{S}'_1$}   & \multicolumn{2}{c|}{$\mathcal{S}'_2$}   & \multicolumn{2}{c|}{$\mathcal{S}'_3$}   \\ \hline
\multicolumn{2}{|c|}{$(1,2)$} & \multicolumn{2}{c|}{$(1,4)$} & \multicolumn{2}{c|}{$(1,3)$} \\ \hline
\multicolumn{2}{|c|}{$(2,1)$} & \multicolumn{2}{c|}{$(2,3)$} & \multicolumn{2}{c|}{$(2,4)$} \\ \hline
\multicolumn{2}{|c|}{$(3,4)$} & \multicolumn{2}{c|}{$(3,2)$} & \multicolumn{2}{c|}{$(3,1)$} \\ \hline
\multicolumn{2}{|c|}{$(4,3)$} & \multicolumn{2}{c|}{$(4,1)$} & \multicolumn{2}{c|}{$(4,2)$} \\ \hline
\end{tabular}
    \caption{{Here, $(m,n)$ represents that  $m^{th}$ local projector  $P_m ^{\mathtt{A}_1\mathtt{A}_2}$ clicked in Alice's lab and $n^{th}$ local projector $P_n^{\mathtt{B}_1\mathtt{B}_2}$ clicked in Bob's lab, when both of them performs Bell basis measurement $\mathcal{M}_\alpha$, where $P^\alpha_i:=\ket{B_i}_\alpha\bra{B_i}$, $\alpha\in\{\mathtt{A}_1\mathtt{A}_2,\mathtt{B}_1\mathtt{B}_2\}$. The table states that for a specific subset, say, $\mathcal{S}'_1$, the joint outcome $(m,n)$ can be anything from the set $\lbrace (1,2), (2,1), (3,4), (4,3)\rbrace $.}  }
    \label{tab1}
\end{table}
\end{center}
\paragraph*{}
\paragraph*{}
\textbf{$(4,3)-$ Bell identifiability:} 
Alice and Bob can successfully identify any subset of three Bell states given from the set $\{\ket{B_i}\}_{i=1}^4$ using only LOCC. The protocol is the following. Each agent (Alice or Bob) performs a projective measurement $\mathcal{M}_\alpha:=\{P^\alpha_i|\sum_{i=1}^4 P^\alpha_i=\mathbb{I}_4\}$, where $P^\alpha_i:=\ket{G_i}_\alpha\bra{G_i}$, $\alpha\in\{\mathtt{A}_1\mathtt{A}_2\mathtt{A}_3,\mathtt{B}_1\mathtt{B}_2\mathtt{B}_3\}$, on the three qubits in their respective labs. Here, $\ket{G_i}$ is any of the $8$ orthogonal three-qubit GHZ states:
$\{\frac{1}{\sqrt{2}}(\ket{000}\pm\ket{111}),\frac{1}{\sqrt{2}}(\ket{001}\pm\ket{110}),\frac{1}{\sqrt{2}}(\ket{010}\pm\ket{101}),\frac{1}{\sqrt{2}}(\ket{100}\pm\ket{011})\}$. Afterwards, the agents can know the identity of the shared subset by tallying their measurement outcomes via classical communication. 

For example, if a subset $\mathcal{S}'_1$ is distributed among the agents then their shared three Bell states can be  either of the following six orders--- $\ket{B_1}_{\mathtt{A}_1\mathtt{B}_1}\otimes \ket{B_2}_{\mathtt{A}_2\mathtt{B}_2}\otimes \ket{B_3}_{\mathtt{A}_3\mathtt{B}_3}$ or, $\ket{B_1}_{\mathtt{A}_1\mathtt{B}_1}\otimes \ket{B_3}_{\mathtt{A}_2\mathtt{B}_2}\otimes \ket{B_2}_{\mathtt{A}_3\mathtt{B}_3}$ or, $\ket{B_2}_{\mathtt{A}_1\mathtt{B}_1}\otimes \ket{B_3}_{\mathtt{A}_2\mathtt{B}_2}\otimes \ket{B_1}_{\mathtt{A}_3\mathtt{B}_3}$ or, $\ket{B_2}_{\mathtt{A}_1\mathtt{B}_1}\otimes \ket{B_1}_{\mathtt{A}_2\mathtt{B}_2}\otimes \ket{B_3}_{\mathtt{A}_3\mathtt{B}_3}$ or, $\ket{B_3}_{\mathtt{A}_1\mathtt{B}_1}\otimes \ket{B_1}_{\mathtt{A}_2\mathtt{B}_2}\otimes \ket{B_2}_{\mathtt{A}_3\mathtt{B}_3}$ or, $\ket{B_3}_{\mathtt{A}_1\mathtt{B}_1}\otimes \ket{B_2}_{\mathtt{A}_2\mathtt{B}_2}\otimes \ket{B_1}_{\mathtt{A}_3\mathtt{B}_3}$. If Alice and Bob ignores the order of the distributed state then such shared state can be written as a convex mixture of each of the six possible orders. Hence, Alice and Bob has to discriminate a set of $4\choose 3$ density matrices $\{ \rho_i \}$, where 
\begin{align}
    \rho_i =\dfrac{1}{6} ~Perm_\mu \oper{B_{\alpha_{i}}}\otimes \oper{B_{\beta_{i}}} \otimes \oper{B_{\gamma_{i}}} , \quad \quad \mu=1,2,..,6, \quad \alpha_i \neq \beta_i \neq \gamma_i \in \{ 1,2,3,4\}
\end{align}
and further we have $Tr(\rho_i \rho_j)=\delta_{ij}$ for all $(i,j)$. The set of states $\{ \rho_i\}$ is said to be LOCC indistinguishable if and only if the support vectors of individual $\rho_i$ are LOCC indistinguishable. Now we must show whether the set of vectors $\{ \ket{B_{\alpha_i}}\ket{B_{\beta_i}}\ket{B_{\gamma_i}}\}_{i=1}^4$ are LOCC indistinguishable or not. One can always write $\ket{B_{\alpha_i}}\ket{B_{\beta_i}}\ket{B_{\gamma_i}}$ as 
\begin{align}
    \ket{B_{\alpha_i}}_{A_1 B_1}\ket{B_{\beta_i}}_{A_2 B_2}\ket{B_{\gamma_i}}_{A_3 B_3}= \dfrac{1}{2\sqrt{2}}\sum_{m_i=1}^{8}  \ket{G_{m_i}}_{A_1 A_2 A_3} \otimes U_{m'_{i} m_i} \ket{G_{m_i}}_{B_1 B_2 B_3}, \quad \forall i\in \{1,2,3,4\} \label{eq1}
\end{align}
where $U_{m'_{i} m_i}$ is local unitary acting on $\ket{G_{m_i}}$ such that $U_{m'_{i} m_i}\ket{G_{m_i}}=\mu_{m'_{i} m_i}\ket{G_{m'_i}}$ where $\mu_{m'_{i} m_i} \in \{ -1, +1\}$.  Now from Eq.$~$(\ref{eq1}) it is clear that if Alice measures her three-qubit in $\{ \ket{G_{m_i}}\}$ basis then for a given outcome $m_i =m_j=m \in \{ 1,2,...,8\}$, where $i\neq j$, the pair of collapsed states at Bob's side, \textit{i.e.} $\{ \ket{G_{m'_i}}, \ket{G_{m'_j}}\}$ must be pairwise orthogonal for all $(i,j) \in \{ 1,2,3,4\}$. Hence, the set of vectors $\{ \ket{B_{\alpha_i}}_{A_1 B_1}\ket{B_{\beta_i}}_{A_2 B_2}\ket{B_{\gamma_i}}_{A_3 B_3}\}_{i=1}^4$ are locally distinguishable which implies that $\{ \rho_i\}_{i=1}^4$ is LOCC distinguishable. 

\begin{center}
\begin{table}[h]
\begin{tabular}{|cl|cl|cl|cl|}
\hline
\multicolumn{2}{|c|}{$\mathcal{S}'_1$}   & \multicolumn{2}{c|}{$\mathcal{S}'_2$}   & \multicolumn{2}{c|}{$\mathcal{S}'_3$} & \multicolumn{2}{c|}{$\mathcal{S}'_4$}  \\ \hline
\multicolumn{2}{|c|}{$(1,4)$} & \multicolumn{2}{c|}{$(1,3)$} & \multicolumn{2}{c|}{$(1,7)$} & \multicolumn{2}{c|}{$(1,8)$} \\ \hline
\multicolumn{2}{|c|}{$(2,3)$} & \multicolumn{2}{c|}{$(2,4)$} & \multicolumn{2}{c|}{$(2,8)$} & \multicolumn{2}{c|}{$(2,7)$} \\ \hline
\multicolumn{2}{|c|}{$(3,2)$} & \multicolumn{2}{c|}{$(3,1)$} & \multicolumn{2}{c|}{$(3,5)$} & \multicolumn{2}{c|}{$(3,6)$} \\ \hline
\multicolumn{2}{|c|}{$(4,1)$} & \multicolumn{2}{c|}{$(4,2)$} & \multicolumn{2}{c|}{$(4,6)$} & \multicolumn{2}{c|}{$(4,5)$} \\ \hline
\multicolumn{2}{|c|}{$(5,8)$} & \multicolumn{2}{c|}{$(5,7)$} & \multicolumn{2}{c|}{$(5,3)$} & \multicolumn{2}{c|}{$(5,4)$} \\ \hline
\multicolumn{2}{|c|}{$(6,7)$} & \multicolumn{2}{c|}{$(6,8)$} & \multicolumn{2}{c|}{$(6,4)$} & \multicolumn{2}{c|}{$(6,3)$} \\ \hline
\multicolumn{2}{|c|}{$(7,6)$} & \multicolumn{2}{c|}{$(7,5)$} & \multicolumn{2}{c|}{$(7,1)$} & \multicolumn{2}{c|}{$(7,2)$} \\ \hline
\multicolumn{2}{|c|}{$(8,5)$} & \multicolumn{2}{c|}{$(8,6)$} & \multicolumn{2}{c|}{$(8,2)$} & \multicolumn{2}{c|}{$(8,1)$} \\ \hline
\end{tabular}
    \caption{{Here, $(m,n)$ represents the joint outcome for the $m^{th}$ projective measurement  $P_m ^{\mathtt{A}_1\mathtt{A}_2\mathtt{A}_3}$ performed by Alice in $\{ \ket{G_m}\}_{m=1}^8$ basis and $n^{th}$ projective measurement $P_n^{\mathtt{B}_1\mathtt{B}_2}\mathtt{B}_3$ of Bob in $\{ \ket{G_n}\}_{n=1}^8$ basis. Here each subset $\mathcal{S}'_i$ is chosen in a way, where $\mathcal{S}'_1= \{ \ket{B_1},\ket{B_2},\ket{B_3}\}$, $\mathcal{S}'_2= \{ \ket{B_1},\ket{B_2},\ket{B_4}\}$, $\mathcal{S}'_3= \{ \ket{B_1},\ket{B_3},\ket{B_4}\}$ and  $\mathcal{S}'_4= \{ \ket{B_2},\ket{B_3},\ket{B_4}\}$. The table states that for a specific subset, say, $\mathcal{S}'_1$, the joint outcome $(m,n)$ can be anything from the set $\lbrace (1,4), (2,3), (3,2), (4,1), (5,8), (6,7), (7,6), (8,5)\rbrace $.}}
\label{tab1}
\end{table}
\end{center}
\paragraph*{}
\textbf{$(4,3)-$ Bell unmarkability:} Although Alice and Bob can identity any given subset of three Bell states out of four but cannot mark the three states $\{ \ket{B_i}\}$ with LOCC only. The proof is very simple and directly follows from \cite{Hayashi}. The local dimension of three given Bell states is $d=8$ whereas, the cardinality for $(4,3)~$ Bell marking is $4 \times 3! =24$ which is greater than the local dimension. Hence, such set is locally unmarkable.   

}

{\section{Appendix B: Proof of Theorem 3}}
\paragraph*{}
{\textbf{Theorem 3.} 
Consider a complete basis set $\mathcal{S}$ of maximally entangled states (MES) in $\mathbb{C}^d \otimes \mathbb{C}^d$. There are $d^2 \choose k$ possible subsets, each containing $k$ distinct states from $\mathcal{S}$. The set $\mathcal{S}$ is ($d^2,k$)-unidentifiable, if ${d^2 \choose k}>d^k$. Moreover, the set of $D$ maximally entangled states ($D<d^2$) will also be $(D,k)$-unidentifiable, provided ${D \choose k}>d^k$.  }
\\
\begin{proof}
{Consider a maximally entangled state in $\mathbb{C}^d\otimes\mathbb{C}^d$ as $\ket{\Gamma}:=\tfrac{1}{\sqrt{d}}\sum_{j=0}^{d-1}\ket{jj}$. A complete basis set of pairwise orthogonal maximally entangled states, $\mathcal{S}=\lbrace |\Gamma_{l}\rangle \rbrace_{l=1}^{d^2}$ can be generated from the {unique} state $\ket{\Gamma}$ as $|\Gamma_{l}\rangle = (\mathbb{U}_l\otimes \mathbb{I})\ket{\Gamma}$, where $\{\mathbb{U}_l\}_{l=1}^{d^2}$ represents a complete set of Hilbert-Schmidt orthogonal unitary operators from $SU(d)$. A subset of $k$ distinct maximally entangled states can be chosen from $\mathcal{S}$ in $\kappa:={d^2 \choose k}$ ways. We denote the subsets as $\mathcal{S}'^{(d^2,k)}_i$, where $i$ runs from $1$ to $\kappa$.} 

Now, the objective is to distinguish the subsets via LOCC. This actually amounts to locally distinguishing corresponding mixed states $\{\rho_i^m\}_{i=1}^\kappa$, where each state $\rho_i^m$ is an equal classical mixture of all possible permutations of the composition of the $k$ pure states from $\mathcal{S}'^{(d^2,k)}_i$. These mixed states are clearly orthogonal to each other because, by construction, each subset contains at least one different element from the other. If the $\rho_i^m$s are distinguishable, then any set of $\kappa$ vectors, each chosen from the support of different $\rho_i^m$s in $\mathbb{C}^{d^k}\otimes\mathbb{C}^{d^k}$, are also distinguishable. However, any $\kappa$ pure maximally entangled states in $\mathbb{C}^{d^k}\otimes\mathbb{C}^{d^k}$ are locally indistinguishable if $\kappa>d^k$ \cite{Hayashi}. Clearly, the set $\{\rho_i^m\}_{i=1}^\kappa$ is locally indistinguishable, otherwise it leads to a contradiction. Evidently, if we consider a set of $D$ ($<d^2$) maximally entangled states, it can be easily shown following a similar approach as above, that the set will be $(D,k)$-unidentifiable for ${D \choose k}>d^k$. This completes our proof. 
\end{proof}

\section{Appendix C: Proof of Theorem 4}
\paragraph*{}
{\textbf{Theorem 4.}\quad Consider a complete basis set $\mathcal{S}$ of GHZ states in $\mathbb{C}^2 \otimes \mathbb{C}^2 \otimes \mathbb{C}^2$. The set $\mathcal{S}$ is locally $(D,k)$-unidentifiable if ${D\choose k}>2^{2k}$ (where
$D\in(2,8]$ and $k\in(1,D)$) 
even when two out of three parties come together in a same lab (collaborate). }
\\

{We start by providing a sketchy version of the proof involving $D=8, k=2$ for the sake of readability. The general version of the proof is given below.} 

{ As previously mentioned that $\mathcal{S}$ is a complete set of orthonormal tripartite $GHZ$ basis defined as $\mathcal{S}:=\{\frac{1}{\sqrt{2}}(\ket{000}\pm\ket{111}),\frac{1}{\sqrt{2}}(\ket{001}\pm\ket{110}),\frac{1}{\sqrt{2}}(\ket{010}\pm\ket{101}),\frac{1}{\sqrt{2}}(\ket{100}\pm\ket{011})\}$. First, note that the set $\mathcal{S}$ has the following property. If any two of the three spatially separated agents sharing an arbitrary state, say $|G_{\alpha}\rangle \in \mathcal{S}$ collaborate then they can transform $|G_{\alpha} \rangle $ to any other state, say $|G_{\beta} \rangle \in \mathcal{S}$ via implementing two-qubit unitary operations on their subsystems.
This property makes $\mathcal{S}$ locally indistinguishable, even if two agents collaborate \cite{nathanson}. Now, if the three agents share a pair of states from $\mathcal{S}$, say $\ket{G_\alpha}\otimes\ket{G_\beta}$ with $\alpha \neq \beta $ then they can transform the pair of shared states to any of the possible $^8P_2=56$ choices from the set $\tilde{S}=\{\ket{G_\alpha}\otimes\ket{G_\beta} | \alpha \neq \beta,~  \forall\{\alpha, \beta\} \in\{1,\cdots,8\}\}$, with any two agents collaborating and implementing four-qubit unitaries on their subsystems of the pair of three-qubit states. Following the arguments in \cite{nathanson}, the set $\{\ket{G_\alpha}\otimes\ket{G_\beta}\}_{\alpha<\beta=1}^8$ is locally indistinguishable even with two collaborating agents, as one can write ${8\choose 2}=26>2^4$.
On the other hand, perfect local identification of the subset of any two states from $\mathcal{S}$ necessarily implies perfect local discrimination of the set of 26 pairwise orthogonal mixed states, $\{\rho_{\alpha\beta}:=\frac{1}{2}(\oper{G_\alpha}\otimes\oper{G_\beta}+\oper{G_\beta}\otimes\oper{G_\alpha})\}_{\alpha<\beta=1}^8$ { with certainty}. But, perfect distinguishability of the set $\{\rho_{\alpha\beta}\}_{\alpha<\beta=1}^8$ ensures perfect discrimination of its support vectors, \textit{i.e.,} discrimination of the set of vectors, $\{\ket{G_\alpha}\otimes\ket{G_\beta}\}_{\alpha<\beta=1}^8$. Hence, the set $\{\rho_{\alpha\beta}\}_{\alpha<\beta=1}^8$ is locally indistinguishable and consequently, $\mathcal{S}$ is locally ($8,2$)-unidentifiable even if any two of the three agents perform any two-qubit operation. This readily rules out the scope of local identification when all three parties are separated. {Now, we discuss the general proof.} 
}

\begin{proof}
Firstly, note that elements of the set $\mathcal{S}$ follow the relation below,
\begin{align}
    \ket{G_{\alpha}} = (\mathbb{I}\otimes\mathbb{U}_{\alpha\beta}) \ket{G_{\beta}},~~\alpha,\beta\in\{1,\ldots,8\},
\end{align}
where, $\mathbb{U}_{\alpha\beta}$ is a unitary acting on two qubits and $\mathbb{I}$ is the identity operator acting on $\mathbb{C}^2$.
Naturally, any subset of $\mathcal{S}$ inherits the same property.

Let $\Bar{\mathcal{S}}:=\{\ket{G_{\alpha}}\}_{\alpha=1}^D$, $2<D\leq8$. There are $\kappa := {D\choose k}$ different ways in which $k$ states can be chosen from $\Bar{\mathcal{S}}$ to form subsets that we label as $\mathcal{S}'_i$, $i=1,\ldots,\kappa$. Locally identifying $\mathcal{S}'_i$ tantamounts to locally distinguishing the set of mixed states $\{\rho_i\}_{i=1}^{\kappa}$ where
\begin{align}
    \rho_i := \frac{1}{k!}\sum_{\mu=1}^{k!}&\ket{\Phi^i_\mu}\bra{\Phi^i_\mu},\\
    \ket{\Phi^i_\mu} = \text{Perm}_\mu \ket{G_{\alpha_1}}&\ket{G_{\alpha_2}}\ldots\ket{G_{\alpha_k}},~\mu=1,\ldots,k!,\nonumber\\\nonumber
    &\alpha_1\neq\alpha_2\ldots\neq\alpha_k\in\{1,\ldots,D\},\\\nonumber
    \text{such that,}~~~~\ket{G_{\alpha_1}},&\ket{G_{\alpha_2}},\ldots,\ket{G_{\alpha_k}}\in\mathcal{S}'_i.\nonumber
\end{align}
Here, $\text{Perm}_\mu \ket{G_{\alpha_1}}\ket{G_{\alpha_2}}\ldots\ket{G_{\alpha_k}}$ means different possible permutations of the state $\ket{G_{\alpha_1}}\ket{G_{\alpha_2}}\ldots\ket{G_{\alpha_k}}$. For example,  $~\text{Perm}_1 \ket{G_1}\ket{G_4}\ket{G_7}=\ket{G_1}\ket{G_4}\ket{G_7}$, $~\text{Perm}_2 \ket{G_1}\ket{G_4}\ket{G_7}=\ket{G_7}\ket{G_4}\ket{G_1}$, 
$\\$
$~\text{Perm}_3 \ket{G_1}\ket{G_4}\ket{G_7}=\ket{G_4}\ket{G_7}\ket{G_1}$ and so on. {Moreover, another property must hold for the set $\{ \rho_i \}_{i=1}^{\kappa}$ given by $Tr(\rho_i \rho_j)=\delta_{ij}$ since the support of $\rho_i$ \textit{i.e.,} $supp(\rho_i)$ is orthogonal to $supp(\rho_j)$ for all $(i,j)$.} We examine the {distinguishability} of $\mathcal{S}'_i$s (equivalently, distinguishability of the $\rho_i$s) when two out of the three qubits of each of the $k$ $\ket{G_{\alpha}}$s are in the same lab. 
Under such condition, it is possible to transform one permutation to any other by implementing local unitaries on the two qubits of each of the $k$ $\ket{G_{\alpha}}$s. Therefore,
\begin{align}
    \ket{G_{\alpha'_1}}\ket{G_{\alpha'_2}}\ldots\ket{G_{\alpha'_k}} \nonumber = (\mathbb{I}\otimes\mathbb{U}_{\alpha'_1\alpha_1})\ket{G_{\alpha_1}}(\mathbb{I}\otimes\mathbb{U}_{\alpha'_2\alpha_2})\ket{G_{\alpha_2}}\ldots(\mathbb{I}\otimes\mathbb{U}_{\alpha'_k\alpha_k})\ket{G_{\alpha_k}}.
\end{align}
Clearly, the set $\{\ket{\Phi^i_\mu}\}_{i=1}^{\kappa}\subset\mathbb{C}^{2^k}\otimes\mathbb{C}^{2^{2k}}$ has the property: $\ket{\Phi^i_\mu}=(\mathbb{I}\otimes\mathbb{U}_{ij,\mu\nu})\ket{\Phi^j_\nu}$ for $i,j=1,\ldots,\kappa$ {and} $\mu,\nu=1,\ldots,k!$ where $\mathbb{U}_{ij,\mu\nu}$s are $2k$-qubit unitaries. So, if $\kappa>2^{2k}$, $\{\ket{\Phi^i_\mu}\}_{i=1}^{\kappa}$ is not perfectly distinguishable \cite{nathanson}. 

On the other hand, if we assume that $\rho_i$s are distinguishable, then so are their supports. This implies that the set $\{\ket{\Phi^i_{\mu}}\}_{i=1}^\kappa$ is also distinguishable, as for each $i$, $\ket{\Phi^i_{\mu}} \in \text{Support}(\rho_i)$. But this leads to contradiction for $\kappa>2^{2k}$. So, the $\mathcal{S}'_i$s are locally {indistinguishable} in each bipartition, provided ${D\choose k}>2^{2k}$, that is, the sets $\mathcal{S}'_i$s are \textit{genuinely} unindentifiable.
\end{proof}

\section{Appendix D: Proof of Theorem 5}
\paragraph*{}

\begin{figure}[h!]
    \centering
    \includegraphics[height=9cm, width=19cm]{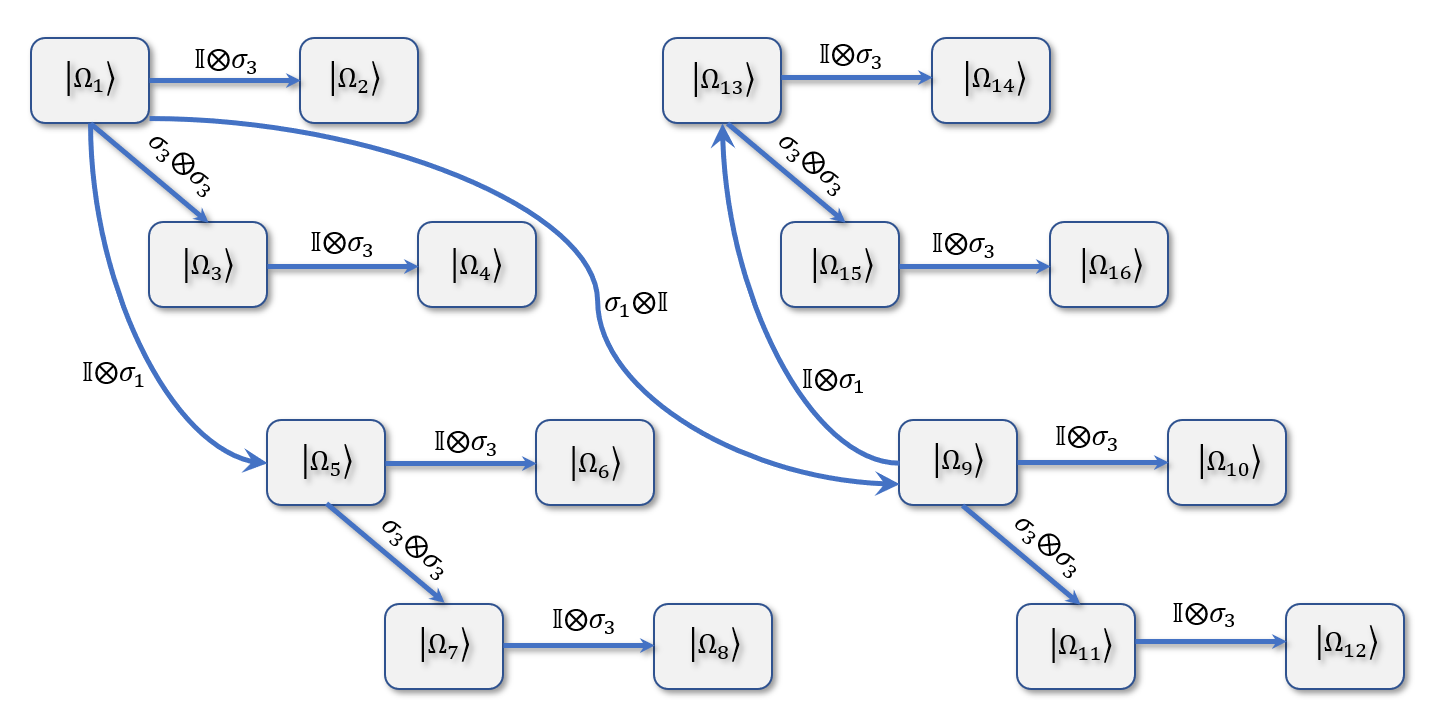}
    \caption{This diagram represents the possible transformation from the state $\ket{\Omega_{1}} \in \mathcal{S}$ to the remaining states $\{\ket{\Omega_{\alpha \neq 1}} \}\in \mathcal{S}$ through a set of unitary operations $\lbrace V_{\alpha \gamma}= \sigma_{\mu}\otimes \sigma_{\nu}\rbrace$, where $\mu, \nu =0,1,2,3$ that satisfies the condition $\ket{\Omega_{\alpha}} =(\mathbb{I}_4 \otimes V_{\alpha \gamma}) \ket{\Omega_{\gamma}}$, where $\mathbb{I}_4$ is an identity operation on $\mathbb{C}^2 \otimes \mathbb{C}^2$. One can trivially find other possible set of unitaries $\{ W_{\alpha \beta}\}$, for which the transformation $\ket{\Omega_{\alpha}} =(\mathbb{I}_2 \otimes W_{\alpha \beta}) \ket{\Omega_{\beta}}$ is satisfied.
    }
    \label{fig:my_label}
\end{figure}

{\textbf{Theorem 5.} \quad Consider a complete basis set $\mathcal{S}$ (set of states given in equation \ref{g4}) of genuinely entangled states in $\mathbb{C}^2 \otimes \mathbb{C}^2 \otimes \mathbb{C}^2 \otimes \mathbb{C}^2$. The set $\mathcal{S}$ is locally $(D,k)$-unidentifiable if ${D\choose k}>2^{3k}$ (where $D\in(2,16]$ and $k\in(1,D)$) in all $1:3$ as well $2:2$ bi-partitions.}
\\

\begin{proof}
{The basic approach for this proof will be similar as of \textbf{Theorem 4} but with a more general picture. So let us write a complete set of pairwise orthogonal four party states as $\mathcal{S}:=\{\ket{\Omega_{\alpha}}\}_{\alpha=1}^{16}$, where
{\small
\begin{subequations}\label{g4}
\begin{align}
    \ket{\Omega_{1,2}} &= \tfrac{1}{2} (\ket{0000}\pm\ket{0111}+\ket{1010}\pm\ket{1101}),\\
    \ket{\Omega_{3,4}} &= \tfrac{1}{2} (\ket{0000}\pm\ket{0111}-\ket{1010}\mp\ket{1101}),\\
    \ket{\Omega_{5,6}} &= \tfrac{1}{2} (\ket{0001}\pm\ket{0110}+\ket{1011}\pm\ket{1100}),\\
    \ket{\Omega_{7,8}} &= \tfrac{1}{2} (\ket{0001}\pm\ket{0110}-\ket{1011}\mp\ket{1100}),\\
    \ket{\Omega_{9,10}} &= \tfrac{1}{2} (\ket{0010}\pm\ket{0101}+\ket{1000}\pm\ket{1111}),\\
    \ket{\Omega_{11,12}} &= \tfrac{1}{2} (\ket{0010}\pm\ket{0101}-\ket{1000}\mp\ket{1111}),\\
    \ket{\Omega_{13,14}} &= \tfrac{1}{2} (\ket{0011}\pm\ket{0100}+\ket{1001}\pm\ket{1110}),\\
    \ket{\Omega_{15,16}} &= \tfrac{1}{2} (\ket{0011}\pm\ket{0100}-\ket{1001}\mp\ket{1110}).
    \end{align}
\end{subequations}} 
At first we need to prove that $\mathcal{S}:=\{\ket{\Omega_{\alpha}}\}_{\alpha=1}^{16}$ is locally indialignstinguishable. Thus we need to show that this set has the following properties: 
\begin{align}
    &|\Omega_{\alpha}\rangle =(\mathbb{I}\otimes W_{\alpha \beta}) ~\ket{\Omega_{\beta}} = (\mathbb{I}\otimes \mathbb{I}\otimes V_{\alpha \gamma})~\ket{\Omega_{\gamma}}, \quad \quad \forall{\alpha \neq \beta \neq \gamma} \in \{ 1,2,3,....,16\}
\end{align}
where $\mathbb{I}$ implies identity operation on a single party, $W_{\alpha \beta}$ implies a joint unitary acting on any three party {or equivalently on any $(1 : 3)$ bi-partition} and $V_{\alpha \gamma}$ implies a joint unitary action on any two party.  Now we will approach in a similar way as the proof of \textbf{Theorem 4.}

Let $\Bar{\mathcal{S}}:=\{\ket{\Omega_{\alpha}}\}_{\alpha=1}^D$, $2<D\leq16$. There are $\kappa := {D\choose k}$ different ways in which $k$ states can be chosen from $\Bar{\mathcal{S}}$ to form subsets that we label as $\mathcal{S}'_i$, $i=1,\ldots,\kappa$. Thus we have to see whether we can locally distinguish the set of mixed states $\{\rho_i\}_{i=1}^{\kappa}$ where
\begin{align}
    \rho_i := \frac{1}{k!}\sum_{\mu=1}^{k!}&\ket{\Phi^i_\mu}\bra{\Phi^i_\mu},\\
    \ket{\Phi^i_\mu} = \text{Perm}_\mu \ket{\Omega_{\alpha_1}}&\ket{\Omega_{\alpha_2}}\ldots\ket{\Omega_{\alpha_k}},~\mu=1,\ldots,k!,\nonumber\\\nonumber
    &\alpha_1\neq\alpha_2\ldots\neq\alpha_k\in\{1,\ldots,D\},\\\nonumber
    \text{such that,}~~~~\ket{\Omega_{\alpha_1}},&\ket{\Omega_{\alpha_2}},\ldots,\ket{\Omega_{\alpha_k}}\in\mathcal{S}'_i.\nonumber
\end{align}
Here, $\text{Perm}_\mu \ket{\Omega_{\alpha_1}}\ket{\Omega_{\alpha_2}}\ldots\ket{\Omega_{\alpha_k}}$ means different possible permutations of the state $\ket{\Omega_{\alpha_1}}\ket{\Omega_{\alpha_2}}\ldots\ket{\Omega_{\alpha_k}}$. 
We examine the {distinguishability} of $\mathcal{S}'_i$s (equivalently, distinguishability of the $\rho_i$s) when either two (2 : 2 bipartition) or three (1 : 3 bipartition) out of the four qubits of each of the $k$ elements from $\{\ket{\Omega_{\alpha_i}}\}$s are in the same lab. Under such condition, it is possible to transform one permutation, say $\mu$ to any other, say $\mu'$ by implementing local unitaries on the two qubits of each of the shared state $\ket{\Omega_{\alpha_i}}$ as 
\begin{align}
    &\ket{\Omega_{\alpha'_1}}\ket{\Omega_{\alpha'_2}}\ldots\ket{\Omega_{\alpha'_k}} \nonumber\\
    = &(\mathbb{I}\otimes W_{\alpha'_1\alpha_1})\ket{\Omega_{\alpha_1}}(\mathbb{I}\otimes W_{\alpha'_2\alpha_2})\ket{\Omega_{\alpha_2}}\ldots(\mathbb{I}\otimes W_{\alpha'_k\alpha_k})\ket{\Omega_{\alpha_k}} \nonumber \\
    =& (\mathbb{I}\otimes \mathbb{I}\otimes V_{\alpha_1 \gamma_1})~\ket{\Omega_{\gamma_1}} (\mathbb{I}\otimes \mathbb{I}\otimes V_{\alpha_2 \gamma_2})~\ket{\Omega_{\gamma_2}} \ldots ~ (\mathbb{I}\otimes \mathbb{I}\otimes V_{\alpha_k \gamma_k})~\ket{\Omega_{\gamma_k}}
\end{align}
Clearly, the support vectors of $\rho_i$ or the set $\{\ket{\Phi^i_\mu}\}_{i=1}^{\kappa}$ is a strict subset of $\mathbb{C}^{2^k}\otimes\mathbb{C}^{2^{3k}}$ as well as of $\mathbb{C}^{2^{2k}}\otimes\mathbb{C}^{2^{2k}}$ and has the property: 
\begin{align}
    \ket{\Phi^{i}_\mu} &=(\mathbb{I}\otimes\mathbb{U}_{ij,\mu\nu})\ket{\Phi^j_\nu}  \quad \quad \forall i,j=1,\ldots,\kappa, \mu,\nu=1,\ldots,k! \nonumber \\
    &= (\mathbb{I} \otimes \mathbb{I}\otimes \mathbb{W}_{il,\mu \delta})\ket{\Phi^l _\delta} \quad \quad \forall i,l=1,\ldots,\kappa, \mu,\delta=1,\ldots,k! 
\end{align}
 where $\mathbb{U}_{ij,\mu\nu}$s are $3k$-qubit unitaries and $\mathbb{W}_{il,\mu \delta}$ are $2k$ qubit unitaries. Thus the local dimension required for implementing the unitary $\mathbb{U}_{ij,\mu\nu}$ is $2^{3k}$ whereas for $\mathbb{W}_{il,\mu\delta}$, it is $2^{2k}$ and we know that $2^{3k}>2^{2k}$ for any $k>0$. Hence, if cardinality $\kappa$ of the set $\{\rho_i \}_{i=1}^{\kappa}$ satisfies the strict condition, $\kappa>2^{2k}$, then $\{\ket{\Phi^i_\mu}\}_{i=1}^{\kappa}$ is not perfectly distinguishable \cite{nathanson}. 

On the other hand, if we assume that $\rho_i$s are distinguishable, then so are their supports. This implies that the set $\{\ket{\Phi^i_{\mu}}\}_{i=1}^\kappa$ is also distinguishable, as for each $i$, $\ket{\Phi^i_{\mu}} \in \text{Support}(\rho_i)$. But this leads to contradiction for $\kappa>2^{2k}$. So, the $\mathcal{S}'_i$s are locally unidentifiable in each bipartition, provided ${D\choose k}>2^{2k}$, that is, the sets $\mathcal{S}'_i$s are \textit{genuinely} unindentifiable.}
\end{proof}

\end{appendix}

\end{document}